\documentclass[12pt]{article}
\usepackage{amsmath, amsthm, amssymb, amsfonts}
\usepackage{url}
\usepackage{caption}
\usepackage{subcaption}
\usepackage[all]{xy}
\usepackage{graphicx}
\usepackage{algorithm, algorithmicx}
\usepackage[noend]{algpseudocode}

\newtheorem{thm}{Theorem}[section]
\newtheorem{lem}[thm]{Lemma}

\theoremstyle{definition}
\newtheorem{defn}[thm]{Definition}
\theoremstyle{remark}
\newtheorem{ex}[thm]{Example}
\newtheorem{rem}[thm]{Remark}

\algrenewcommand\algorithmicrequire{\textbf{Precondition:}}
\algrenewcommand\algorithmicensure{\textbf{Postcondition:}}

\newcommand{\ord}{\operatorname{ord}}
\newcommand{\into}{\hookrightarrow}
\def\d{\partial} 
\def\od{\stackrel{\mathrm{def}}{=}}
\def\im{\operatorname{im}}
\def\ker{\operatorname{ker}}

\usepackage{multicol}
\usepackage{color}
\definecolor{orangered}{rgb}{0.85,.3,0}

\begin{document}

\begin{center}
{\Large Clique topology of real symmetric matrices}\\
Supplementary Text for:\\ {\it Clique topology reveals intrinsic geometric structure in neural correlations}\\
{Carina Curto, Chad Giusti, and Vladimir Itskov}\\
{Feb 21, 2015}
\end{center}

\tableofcontents

\bigskip

\section{Introduction}\label{S:intro}

The purpose of this supplement is to provide a more complete account of the mathematics underlying our analyses in the main text.  In particular, the {\it order complex} and {\it clique topology} are described more precisely here.  The order complex of a matrix is analogous to its Jordan Form, in that it captures features that are invariant under a certain type of matrix transformation.  Likewise, the clique topology of a matrix is analogous to its eigenvalue spectrum, in that it provides a set of invariants that can be used to detect structure.  While the Jordan Form and eigenvalue spectrum are invariant under linear  change of variables, the order complex and clique topology are invariant under monotonic transformations of the matrix entries.

Seeking quantities that are invariant under linear coordinate transformations is natural in physical applications, where measurements are often performed with respect to an arbitrary basis, such as the choice of $x$, $y$ and $z$ directions in physical space.  
In contrast, measurements in biological settings are often obtained as nonlinear (but monotonic) transformations of the underlying ``real'' variables, while the choice of basis is meaningful and fixed.  For example, basis elements might represent particular neurons or genes, and measurements (matrix elements) could consist of pairwise correlations in neural activity, or the co-expression of pairs of genes.  Unlike change of basis, these transformations are of the form
$$L_{ij} = f(M_{ij}),$$
where $f$ is a nonlinear, but monotonically increasing function that is applied to each entry of $M$.  The Jordan Form of a matrix, and its spectrum, may be badly distorted by such transformations; it also discards basis information which may be meaningful and should be preserved.  

Given a symmetric, $N \times N$ matrix that reflects correlations or similarities between $N$ entities (such as neurons, imaging voxels, etc.), we have two basic questions:
\begin{enumerate}
\item[Q1.] Is the matrix a monotonic transformation of a random or geometric\footnote{Recall from the main text that a {\it geometric} matrix refers to a matrix of (negative) Euclidean distances among random points in $\mathbb R^d$ .} matrix?   
\item[Q2.] Can we distinguish between these two possibilities, without knowing $f$? 
\end{enumerate}
Perhaps surprisingly, information sufficient to answer these questions is contained in the ordering of matrix entries, and is encoded in its {\it order complex}, to be described in the next section.  To extract the relevant features, we compute certain topological invariants of the order complex, which we refer to as the \emph{clique topology} of the matrix. The motivation for this choice stems from recent mathematical results by M. Kahle \cite{Kahle2009}, describing the clique topology of random symmetric matrices asymptotically (for large $N$); and our own computational results, showing that random and ``generic'' Euclidean distance matrices can be readily distinguished using clique topology for $N \sim 100$.

We have made an effort to keep these explanations self-contained, but details of how certain computations are performed have been left to the references for the sake of brevity.  Standard material from algebraic topology  \cite{Hatcher} is described in a minimal fashion, with an emphasis on homology of clique complexes.  The reader is expected to be familiar with linear algebra.

\subsubsection*{Comparison to prior applications in biology}

Topological data analysis has previously been used in biological applications to identify individual persistent cycles that may have meaningful interpretation \cite{gap, SinghMemoliIshkhanovSapiroCarlssonRingach2008, NicolauLevineCarlsson2011, DabaghianMemoliFrankCarlsson2012, ChanCarlssonRabadan2013, Chen2014}.  In contrast, our approach relies on the {\it statistical properties} of cycles, 
as captured by Betti curves, in order to detect geometric structure (or randomness) in symmetric matrices.
In particular, the relevant space from which the data points are sampled may not possess any meaningful persistent cycles, as in the square box environment covered by place fields.  The background Euclidean geometry, however, has a strong effect on the statistics of cycles, enabling detection of geometric structure and providing a sharp contrast to Betti curves of random matrices with i.i.d. entries.

\section{The order complex}


Recall that a function $f: \mathbb{R} \to \mathbb{R}$ is said to be {\em monotonically increasing} if $f(x) > f(y)$ whenever $x > y$.
Let $f: \mathbb{R} \to \mathbb{R}$ be a monotonically increasing function. For any real-valued matrix $M$, we define the matrix $f \cdot M$ by
\begin{equation*}
(f \cdot M)_{ij} = f(M_{ij}).   
\end{equation*}
 Note that this action preserves the ordering of matrix entries.  That is, if $L = f \cdot M$, then all pairs of off-diagonal entries, $(i,j)$ and $(k,\ell),$ satisfy:
 $$L_{ij} < L_{k\ell} \; \Leftrightarrow \; M_{ij} < M_{k\ell}.$$
 
 Equivalence classes of matrices can thus be represented by integer-valued matrices that record the ordering of off-diagonal entries (and carry no information on the diagonal).  Figure 1 shows three matrix orderings for $N = 5$.
  \begin{figure}[!h]\begin{center}
\includegraphics[width=3.5in]{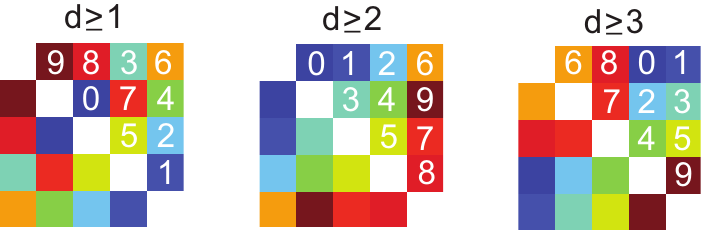}
\caption{Three matrix orderings, reproduced from Figure 2a in the main text.}
\end{center}
\vspace{-.2in}
\end{figure}
For a given symmetric matrix $M$, we denote the representative matrix ordering by $\widehat{M}$, where
 $$\widehat{M}_{ij} = |\{(k,\ell) \mid 0 < k < \ell \leq N \; \text{and} \; M_{k\ell} < M_{ij}\}|$$
 simply counts the number of upper-triangular entries of $M$ that are smaller than $M_{ij}$ for $i \neq j$, while the diagonal entries of $\widehat{M}$ are left undefined ($| \cdot |$ denotes the size of the set).  If $M_{ij}$ is the smallest off-diagonal entry, then $\widehat{M}_{ij} = 0$; if $M_{ij}$ is the largest matrix entry, and all upper-triangular entries are distinct, then $\widehat{M}_{ij} = {N \choose 2} -1.$   With this notation, we have:
  
 \begin{lem}\label{lemma:ordercomplex}
$\widehat{L} = \widehat{M}$ if and only if there exists a monotonically increasing function $f$ such that $L = f \cdot M$.
 \end{lem}

\begin{proof}
$(\Leftarrow)$ is obvious, since the action of $f$ preserves the ordering of matrix entries.  $(\Rightarrow)$ One can construct $f:\mathbb{R} \to \mathbb{R}$ by setting $f(M_{ij}) = L_{ij}$ for each off-diagonal entry, and interpolating monotonically (e.g., linearly).  Since we assume $\widehat{L} = \widehat{M}$, this function is monotonically increasing and well-defined.
\end{proof}

 In order to analyze the information present in the ordering of entries for an $N \times N$ symmetric matrix, it is useful to represent it as a sequence of nested simple graphs. Recall that a \emph{simple graph} $G$ is a pair ($[N]$, $E$), where $[N] = \{1, 2, \dots , N\}$ is the ordered set of \emph{vertices}, and $E$ is the set of \emph{edges}.   Each edge is undirected and connects a unique pair of distinct vertices (no self-loops).
 We will use the notation $(ij) \in G$ to indicate that the edge corresponding to vertices $i,j$ is in the graph.

\vspace{.2in}
\begin{figure}[!h] 
\begin{center}
\includegraphics[width=6in]{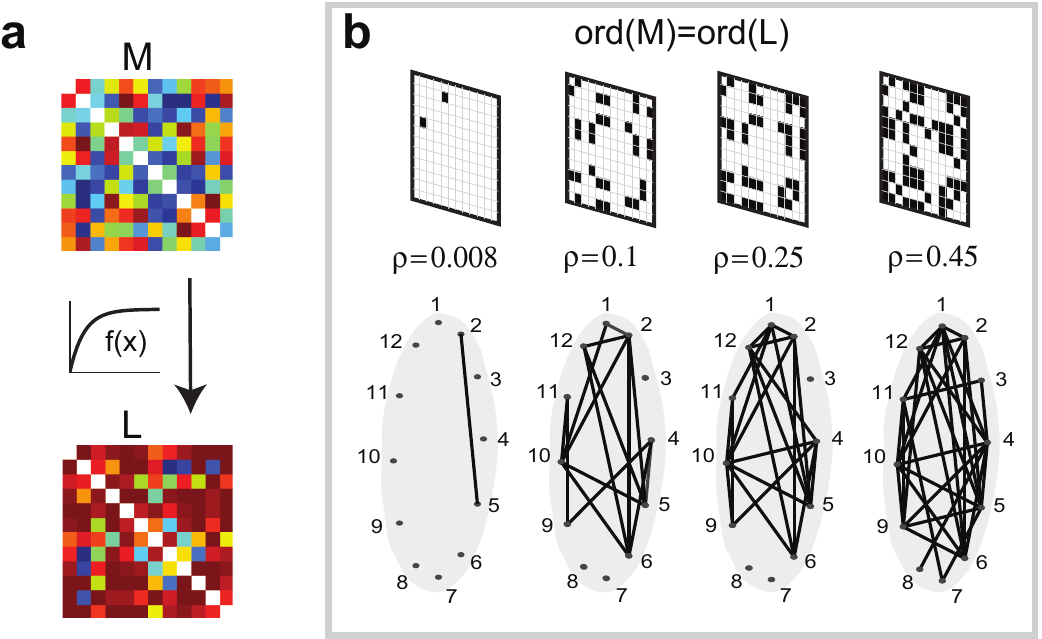}
\caption{Selected graphs in an order complex, adapted from Figure 1 in the main text.}
\end{center}
\vspace{-.2in}
\end{figure}

 \begin{defn}
 Let $M$ be a real symmetric matrix with matrix ordering $\widehat{M}$, and let $p = \max_{i < j} \widehat{M}_{ij}$.  The \emph{order complex} of $M$, denoted $\ord(M)$, is the sequence of graphs 
 $$G_0 \subset G_1 \subset \cdots \subset G_{p+1},$$ 
 such that 
 \begin{equation*}
 (ij) \in G_r \;\; \Leftrightarrow \;\; \widehat{M}_{ij} > p - r \;\; \text{for each} \;\; r = 0, \ldots, p+1.
 \end{equation*}
 \end{defn}
Note that $G_0$ has no edges, $G_1$ contains only the edge $(ij)$ corresponding to the largest off-diagonal entry of $M$, and subsequent graphs are each obtained from the previous one by adding an additional edge for each next-largest entry until we reach the complete graph, $G_{p+1}$.  A portion of an order complex is illustrated in Figure 2.  It is clear from the definition that:
 $$\ord(L) = \ord(M)\; \Leftrightarrow\; \widehat{L} = \widehat{M}.$$

Because of Lemma \ref{lemma:ordercomplex}, the order complex $\ord(M)$ captures all features of $M$  that are preserved under the action of monotonically increasing functions.

\section{Clique topology}

We are now ready to introduce {\em clique topology}, a tool for extracting invariant features of a matrix from the ordering of matrix entries.
We begin by describing the clique topology of a single graph $G$, by which we simply mean the homology of its {\it clique complex}:
$$H_i(X(G), \mathbf{k}),$$
where $\mathbf{k}$ is a field (more on the field in section~\ref{sec:chains}).
The clique complex, $X(G),$ is defined in section~\ref{sec:clique-complex}; while the simplicial homology groups, $H_i(X(G), \mathbf{k}),$ will be defined in section~\ref{sec:homology}.  We refer to these invariants as \emph{clique} topology in order to indicate that we are measuring topological features of the organization of cliques in the graph, rather than the usual topology of the graph.

We summarize the information present in clique topology via a set of {\em Betti numbers}, $\beta_i(X(G)),$ which are the ranks of the corresponding homology groups:
$$\beta_i(X(G)) \od \mathrm{rank}\; H_i(X(G), \mathbf{k}).$$
The clique topology of a symmetric matrix $M,$ with order complex $G_0 \subset G_1 \subset \cdots \subset G_{p+1}$,
is reflected in the sequences of Betti numbers $\beta_i(X(G_r))$, computed for various dimensions $i = 0, 1, 2, \ldots,$ and for each graph $G_r$ in $\ord(M)$ (see section~\ref{sec:clique-top-across}).

The reader familiar with homology of simplicial complexes, including clique complexes, should feel free to skip the next few sections and proceed directly to section~\ref{sec:clique-top-across}, where we define {\it Betti curves}.

\subsection{The clique complex of a graph}\label{sec:clique-complex}

Recall that a \emph{clique} in a graph $G$ is an all-to-all connected collection of vertices in $G$. An \emph{$m$-clique} is a clique consisting of $m$ vertices. Note that if $\sigma$ is a clique of $G$, then all subsets of $\sigma$ are also cliques. 

 \begin{defn}
 Let $G$ be a graph with $N$ vertices. The \emph{clique complex} of $G$, denoted $X(G)$, is the set of all cliques of $G$:
 $$X(G) = \{\sigma \subset [N] \mid \sigma \text{ is a clique of } G\}.$$
 We write $X_m(G)$ for the set of $(m+1)$-cliques of $G$.
  \end{defn}
  
 The shift in index reflects the ``dimension'' of a clique, when the clique complex is represented geometrically. If we think of the vertices of the graph $G$ as embedded generically in a high-dimensional space, each clique represents the simplex given by the convex hull of its vertices. For example, the convex hull of two vertices is a 1-dimensional edge, for three vertices we obtain a 2-dimensional triangle, and four vertices yields a 3-dimensional tetrahedron.  Thus, cliques in $X_m(G)$ consist of $m+1$ vertices, but represent $m$-dimensional simplices.  
 
 The \emph{boundary} of a clique $\sigma \subseteq G$ is the collection of subcliques $\tau \subset \sigma$ which have one fewer vertex.  This corresponds to the set of lower-dimensional simplices that comprise the boundary of the simplex defined by $\sigma$ (Figure 3b).
  
The \emph{homology} of a clique complex $X(G)$, to be defined in section~\ref{sec:homology}, is a measurement of relationships among the cliques in $G$. Intuitively, homology counts \emph{cycles} in the clique complex, a higher-dimensional generalization of the notion of cycles in a graph (Figure 3a). A collection of cliques forms a cycle if their boundaries overlap so as to ``cancel'' one another (Figure 3b). We also wish avoid double-counting cycles which are in our geometric sense equivalent. In particular, two cycles are considered equivalent if one can be deformed into the other without leaving the clique complex (Figure 3c).  Alternatively, if we can ``combine'' two cycles to form a third (Figure 3d), we should not detect their concatenation as a new independent cycle. 

\begin{figure}[!h]
\begin{center}
\includegraphics[width=5.6in]{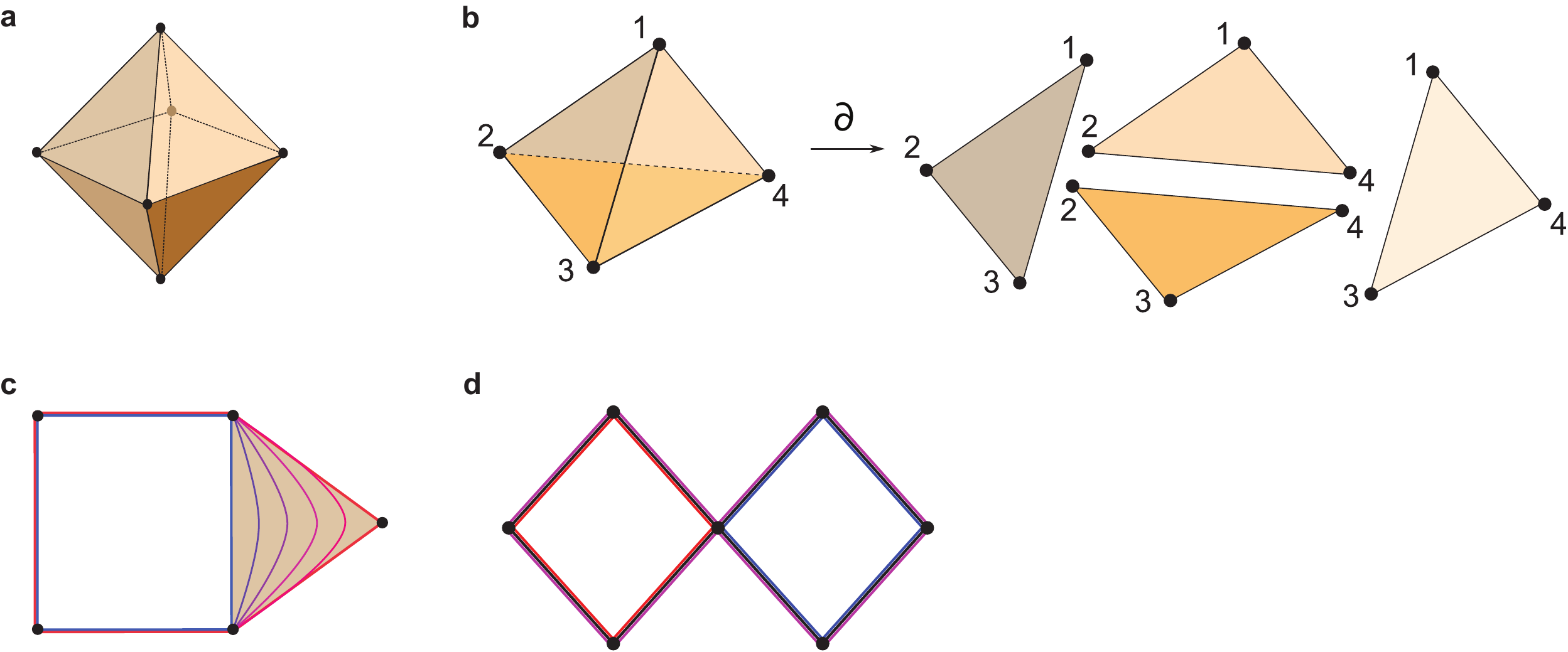}
\caption{Illustrations of homology ideas.}
\end{center}
\end{figure}

 \subsection{Chains and boundaries}\label{sec:chains}

To make the above notion of ``cancellation'' of boundaries precise (and computable), one introduces linear combinations of cliques, called \emph{chains}. 
 Given a set of cliques $\sigma_1,\ldots,\sigma_\ell \in X(G)$, one can form a vector space consisting of formal linear
combinations of cliques with coefficients in a field $\mathbf{k}$:
$$\sum_{i = 1}^\ell a_i c_{\sigma_i}, \;\;\; \text{where}\;\;  a_i \in \mathbf{k},$$
and $c_{\sigma_i}$ denotes the basis element corresponding to the clique $\sigma_i$.  
To define chain groups, one considers linear combinations of cliques of the same size.
Recall that $X_m(G)$ denotes the set of $(m+1)$-cliques of $G$.   
 
  \begin{defn}
The \emph{$m$-th chain group} of $X(G),$ with coefficients in $\mathbf{k},$ is the $\mathbf{k}$-vector space: 
$$C_m(X(G); \mathbf{k}) \od \left\{ \sum_{i = 1}^\ell a_i c_{\sigma_i} \mid \sigma_i \in X_m(G) \text{ and } a_i \in \mathbf{k} \text{ for each } i = 1,\ldots,\ell \right\}.$$
  \end{defn}
  \noindent As we will always be working with coefficients in an aribitrary field $\mathbf{k}$,\footnote{
For readers uncomfortable with the notion of a general field $\mathbf{k}$, it is relatively harmless to substitute $\mathbb{R}$ or $\mathbb{Q}$ for $\mathbf{k}$ for the remainder of the discussion. One should keep in mind, however, that the actual computations typically take place with $\mathbf{k}$ a finite field, $\mathbb{Z}/p\mathbb{Z}$.  This can have an effect on the result: in such a field, one can add a boundary to itself a finite number of times and get zero, creating ``extra'' cycles -- called {\em torsion} cycles -- that would not be present over $\mathbb{R}$. These extra cycles measure aspects of the clique complex that are not relevant to our purposes.  In our software we have chosen the field to be $\mathbb{Z}/2\mathbb{Z},$ but this choice is somewhat arbitrary and not important.  So long as all computations are done using the same field, comparing the resulting homology groups across different graphs is entirely valid.
}
 we will omit it from the notation and write $C_m(X(G))$ instead of $C_m(X(G); \mathbf{k})$.Note that $C_0(X(G))$ consists of formal linear combinations of $1$-cliques (vertices), $C_1(X(G))$ of $2$-cliques (edges), $C_2(X(G))$ of $3$-cliques (triangles), and so on.  

 The boundaries of cliques can also be described algebraically, allowing this notion to be extended to chains.
  If $\sigma = \{v_i\}_{i=0}^m$ is an $m$-clique of $G$, we use the notation 
  $$c_\sigma = c_{v_0v_1\dots v_m},$$ 
  where $v_0 < v_1 < \ldots < v_m$ (recall that each $v_i \in [N]$).  Consistent ordering is important because it affects the signs in the boundary map.  Given a sequence of vertices $v_0v_1\dots v_m$, we denote by $v_0v_1\dots \hat{v}_i \dots v_m$ the sequence obtained by omitting the element $v_i$.  Note that for each $\sigma \in X_m(G)$, the element $c_\sigma = c_{v_0v_1\dots v_m}$ is a basis element of the vector space $C_m(X(G))$.

  \begin{defn}\label{D:boundary}
  The \emph{boundary} map $\d_m : C_m(X(G)) \to C_{m-1}(X(G))$, for $m > 0,$ is given on basis elements $c_{v_0v_1 \dots v_m}$ by
  \begin{equation*}
  \d_m(c_{v_0v_1 \dots v_m}) = \sum_{i=0}^m (-1)^{i}c_{v_0v_1\dots \hat{v}_i\dots v_m},
  \end{equation*}  
  and is extended via linearity to general chains; i.e. $\d_m(\sum_j a_jc_{\sigma_j}) = \sum a_j \d_m(c_{\sigma_j})$. The map $\d_0$ is defined to be the zero map.
  \end{defn}
  
Recall that in the geometric picture, an $(m+1)$-clique corresponds to a $m$-dimensional simplex, and the boundary of this simplex is the set of $m$-cliques comprising its $(m-1)$-dimensional {\it facets} -- that is, all subcliques on one fewer vertex.
We have thus defined the boundary of a chain in $C_m(X(G))$ in a fashion consistent with our geometric understanding: as a formal sum of chains in $C_{m-1}(X(G))$, corresponding to simplices that are one dimension lower (see Figure 3b).  Note that signs are assigned to the elements of this formal sum to indicate the \emph{orientation} of cliques, which will be critical for obtaining the desired ``cancellation'' of boundaries  (see Remark \ref{R:orientation} for details). 

\begin{ex} Suppose $\sigma, \tau \in X_2(G)$ are cliques on vertices $\{1,2,3\}$ and $\{1, 2, 4\}$ respectively. The boundary of the 2-chain $c_\sigma - c_\tau \in C_2(X(G))$ is
\begin{eqnarray*}
\d_2(c_{123} - c_{124}) &=& \d_2(c_{123}) - \d_2(c_{124}) \\
&=& \left(c_{23} - c_{13} + c_{12}\right) - \left(c_{24} - c_{14} + c_{12}\right) \\
&=& c_{23} - c_{13} - c_{24} + c_{14}.
\end{eqnarray*}
The cancellation of $c_{12}$ reflects the fact that the clique $\{1,2\}$ appears twice in the boundary of $c_\sigma - c_\tau$, with \emph{opposite} orientation.
Note also that applying $\d_1$ to the resulting $1$-chain yields
\begin{eqnarray*}
\d_1(\d_2(c_{123} - c_{124})) &=& \d_1(c_{23} - c_{13} - c_{24} + c_{14} )\\
&=& (c_3 - c_2) - (c_3 - c_1) - (c_4 - c_2) + (c_4 - c_1)\\
&=& 0.
\end{eqnarray*}
\end{ex}
In fact, it is straightforward to check from the definition that the composition of two subsequent boundary maps always yields $0$.  In other words,

\begin{lem} \label{L:chaincx} 
For any $m>0$, $\d_{m} \circ \d_{m+1} = 0$.
\end{lem}

\begin{rem}\label{R:orientation}
The \emph{orientation} of cliques can be {\it positive} or {\it negative}. The vertices of a clique $c_{v_0v_1 \dots v_m} \in X_m(G)$  have a \emph{canonical} ordering induced by the usual ordering of the vertices $[N]$ of $G$. We define the canonical ordering to have positive orientation for each clique. Any other ordering can be obtained as a permutation of the canonical ordering, and the resulting ordering is positive or negative according to the sign of the permutation. For example, $c_{124}$ has a positive orientation, while $c_{214}$ is negatively oriented.  When we compute the boundary of a clique $c_\sigma$ in Definition \ref{D:boundary}, the signs arise as a result of the \emph{induced orientation} on the boundary cliques.  The result of taking all cliques on the boundary is the signed sum we obtain in Definition \ref{D:boundary}. \end{rem}

 \subsection{Homology of a clique complex}\label{sec:homology}

 For a given graph $G$, the chain groups $C_m(X(G))$ can be strung together to form a {\em chain complex}:
 \begin{equation*}
  \xymatrix{
  0 \ar[r]^-{\d_{k+1} = 0}& C_k(X(G)) \ar[r]^-{\d_{k}}& C_{k-1}(X(G)) \ar[r]^-{\d_{k-1}}& \dots \ar[r]^-{\d_{2}} & C_1(X(G))\ar[r]^-{\d_{1}}&C_0(X(G))\ar[r]^-{\d_0 = 0}&0,
  }
  \end{equation*}
The zeroes at either end of the complex represent the zero-dimensional $\mathbf{k}$-vector space, and the maps at each end are necessarily the zero map.  
 
    If a chain is in the kernel of the boundary map, it is because the (oriented) boundaries of its constituent cliques cancel one another. This is precisely the desired notion of a \emph{cycle}, so the set of $m$-cycles is exactly $\text{ker}(\d_m)$; in particular, $1$-cycles correspond to the usual notion of cycles in a graph.   Note also that any chain in $C_m(X(G))$ which forms the boundary of a clique in $X_{m+1}(G)$ is itself a cycle, so its own boundary should be zero.   This is reflected in the fact that $\d_m \circ \d_{m+1} = 0$ (Lemma~\ref{L:chaincx}).  In particular, 
    $$\im \d_{m+1} \subset \ker \d_m.$$
  
When we are counting cycles for homology, we do not want to consider those which arise as boundaries of chains, as these are ``filled in."  For example, the two clique complexes in Figure 4 should have the same number of homology $1$-cycles.  In Figure 4b, we do not wish to count the chain $c_{23}+c_{35}-c_{25} \in C_1(X(G))$ as a $1$-cycle because it is the boundary of a clique, $c_{235} \in C_2(X(G))$.

\begin{figure}[!h]
\begin{center}
\includegraphics[width=3.5in]{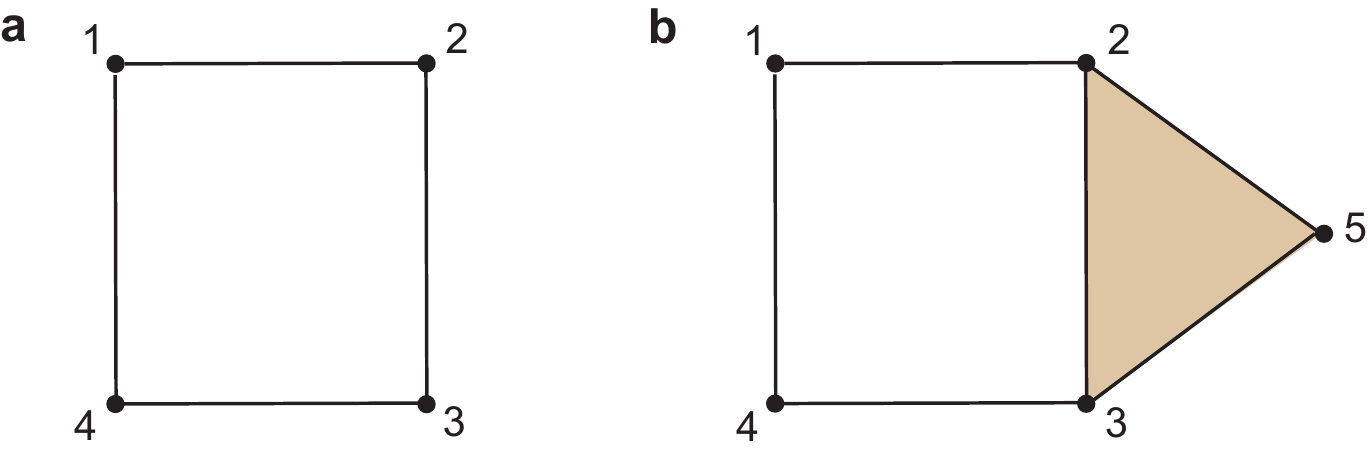}
\caption{Two clique complexes for graphs on 4 and 5 vertices.}
\end{center}
\end{figure}

\noindent In order to eliminate cycles that are boundaries of higher-dimensional cliques, one computes quotient vector spaces, $\ker(\d_m)/\im(\d_{m+1}).$
  
  \begin{defn}
  The $m$-th {\em homology group} of $X(G)$ with coefficients in $\mathbf{k}$ is the quotient space 
  \begin{equation*}
  H_m(X(G); \mathbf{k}) \od \frac{\text{ker}(\d_m)}{\text{im}(\d_{m+1})}.
  \end{equation*}
  \end{defn}
  
\noindent As with chain groups, we will omit the field from our notation and write simply $H_m(X(G))$. 
  
Observe that the zeroth homology group is special: since $\d_0 = 0$, its kernel is always $C_0(X(G))$. The quotient $\ker(\d_0)/\im(\d_{1})$ thus identifies vertices which are connected to one another, so that $H_0(X(G))$ is a vector space whose basis can be chosen to correspond to the \emph{connected components} of $G$.

\begin{ex} \label{E:square} Let $G$ be the graph on four vertices in Figure 4a.   The kernel of the boundary map $\d_1: C_1(X(G)) \to C_0(X(G))$ is the one-dimensional space spanned by $\sigma = c_{12} + c_{23} + c_{34} - c_{14}$.  Indeed, $\d_1(\sigma) = (c_2-c_1)+(c_3-c_2)+(c_4-c_3)-(c_4-c_1) = 0.$  Since there are no cliques of size greater than $2$, $C_2(X(G)) = 0$ and hence $\d_2 = 0$.   It follows that $H_1(X(G))$ is precisely the one-dimensional vector space spanned by $\sigma$.  Furthermore, since $C_1(X(G))$ has dimension $4$ and $\ker \d_1$ has dimension $1$, it follows that $\im \d_1$ has dimension 3.  We can thus deduce that $H_0(X(G))$ is also one-dimensional, consistent with the fact that $G$ has just one connected component.

Next, consider the graph $G'$ on five vertices in Figure 4b. This graph has been obtained from $G$ by ``attaching'' the clique $\{2,3,5\}$. The kernel of $\d_1$ is now 2-dimensional, and is spanned by both $\sigma$ and a new cycle, $\tau = c_{23}+c_{35}-c_{25}$.  However, $\tau \in \im \d_2$, so we find that $H_1(X(G'))$ continues to be one-dimensional, consistent with our intuition that $G$ and $G'$ both have just one cycle that has not been ``filled in'' by cliques.
\end{ex}

\begin{figure}[!h]
\begin{center}
\includegraphics[width=5in]{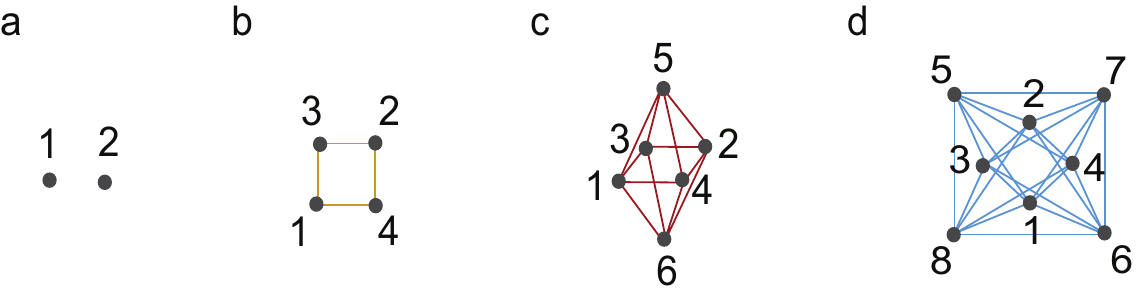}
\caption{Cross-polytopes generate the minimal clique complexes which produce homology in each dimension.  Adapted from Figure 1 of the main text.}
\end{center}
\end{figure}
\vspace{-.2in}

\begin{ex} \label{E:mincycle} The smallest example of a graph $G_m$ whose clique complex has non-trivial $m$-th homology group is the 1-skeleton of the $(m+1)$-dimensional cross-polytope (Figure 5). Such a graph can be built inductively starting from the graph $G_0$ (Figure 5a), having just two vertices and no edges.
To obtain $G_1$ from $G_0$, we attach two new vertices and include all edges between the new vertices and the vertices of $G_0$ (Figure 5b).  More generally, 
to obtain $G_i$ from $G_{i-1}$ we attach two new vertices and all edges between these new vertices and those of $G_{i-1}$. Thus, we obtain $G_2$ (Figure 5c) and $G_3$ (Figure 5d), which give minimal examples of graphs whose clique complexes have a non-trivial homology $2$-cycle and $3$-cycle, respectively.
\end{ex}

 A useful characterization of the clique topology of a graph is obtained by simply tracking the dimensions of the homology groups.   This is done via the so-called Betti numbers.
  
  \begin{defn}
  The \emph{$m$-th Betti number} of $X(G)$, denoted $\beta_m(X(G))$, is the rank of $H_m(X(G); \mathbf{k})$ as a $\mathbf{k}$-vector space.
  \end{defn}
      
  While this information discards the identities of individual cycles, it is well-suited to statistical methods as it reduces the clique topology of a graph to a sequence of integers. 

\subsection{Clique topology across the order complex}\label{sec:clique-top-across}
  
We now turn our attention to the clique topology of all graphs in the order complex at once. For a matrix $M$, the Betti numbers of the graphs in $\ord(M)$ are collected as follows.
  
  \begin{defn}
  Let $M$ be a real symmetric matrix and $\ord(M) = (G_0 \subset G_1 \subset G_2 \subset \cdots \subset G_{p+1})$ its order complex, where $p = \max_{i < j} \widehat{M}_{ij}$. The \emph{$m$-th Betti curve} of $M$ is the sequence of numbers $\{\beta_m(\rho_r)\}_{r=1}^{p+1},$ where $\rho_r$ is the edge density of the graph $G_r$, and 
  $$\beta_m(\rho_r) \od \operatorname{rank} H_m(X(G_r)).$$ 
  As the matrix $M$ will be clear from context, we omit it from the notation.
  \end{defn}
  
  While each Betti curve is a discrete sequence, we can think of it as being a piecewise constant function.   
 To simplify comparison, we consider as a summary statistic the integral of the entire Betti curve.  We call this the \emph{$m$-th total Betti number} of the matrix $M$, given by
 \begin{equation*}
 \bar{\beta}_m (M) \od \sum_{r=1}^{p+1}\beta_m(\rho_r) \Delta \rho_r = \int_{0}^{1} \beta_m(\rho) \;\text{d}\rho, 
 \end{equation*}
 where $\Delta\rho_r$ is the change in edge density between $G_r$ and $G_{r-1}$.\footnote{This measurement, $\bar{\beta}_m(M)$, also appears as the first element in the basis for the ring of algebraic functions on the collection of all persistence structures described in \cite{AdcockCarlssonCarlsson}.}  Typically,  $\Delta \rho_r = 1/{N \choose 2}$, which is the change in density after adding a single edge.  As we will see, the $\bar{\beta}_m$ alone can distinguish between a random symmetric matrix, drawn from a distribution with i.i.d. entries, and a geometric matrix, which arises from distances between a set of randomly-distributed points in Euclidean space.  Thus, we can use the total Betti number to test the hypotheses that a matrix is random or geometric.

\section{Clique topology of random and geometric matrices}\label{S:families}

In order to interpret the results of computing clique topology for matrices of interest, we need suitable null models for comparison.   This brings us back to our motivating questions Q1 and Q2 from section~\ref{S:intro}.  Can we use clique topology to reject the hypothesis that a given matrix is random or geometric?  This will be possible if matrices in these categories have stereotyped Betti curves.  In this case, 
it can be shown that a matrix with a substantially different Betti curve is unlikely to have come from the given null model distribution, and a $p$-value can be assigned to quantify the significance.

Because clique topology depends only on $\ord(M)$, it suffices to describe the distributions of order complexes we obtain for random and geometric matrices.  In both families, the details of the Betti curves change with $N$; however, we find that their large-scale features are robust once $N > 50$.  This means Betti curves can indeed be used to reject these models.
  
The distribution of \emph{random order complexes} arises by sampling a matrix ordering $\widehat{M}$ from the uniform distribution on all such orderings.  For $N \times N$ symmetric matrices with distinct entries $\{M_{ij}\}_{i<j}$, this can be achieved by sampling permutations of $\{0,\ldots,{N \choose 2}-1\}$ uniformly at random.  Equivalently, the matrix can be chosen with i.i.d. entries drawn from any continuous distribution, or by shuffling the elements of a given matrix with distinct off-diagonal entries. Thus, in a graph $G_\rho$ of $\ord(M)$, each edge has independent probability $\rho$ of appearing.  In other words, the graphs in the order complex are a nested family of Erd\"os-R\'enyi random graphs. The clique topology of such complexes is relatively well understood from a theoretical perspective \cite{Kahle2009}, with highly stereotyped, unimodal Betti curves as illustrated in Figure 6a.

\begin{figure}[!h]
\begin{center}
\includegraphics[width=6in]{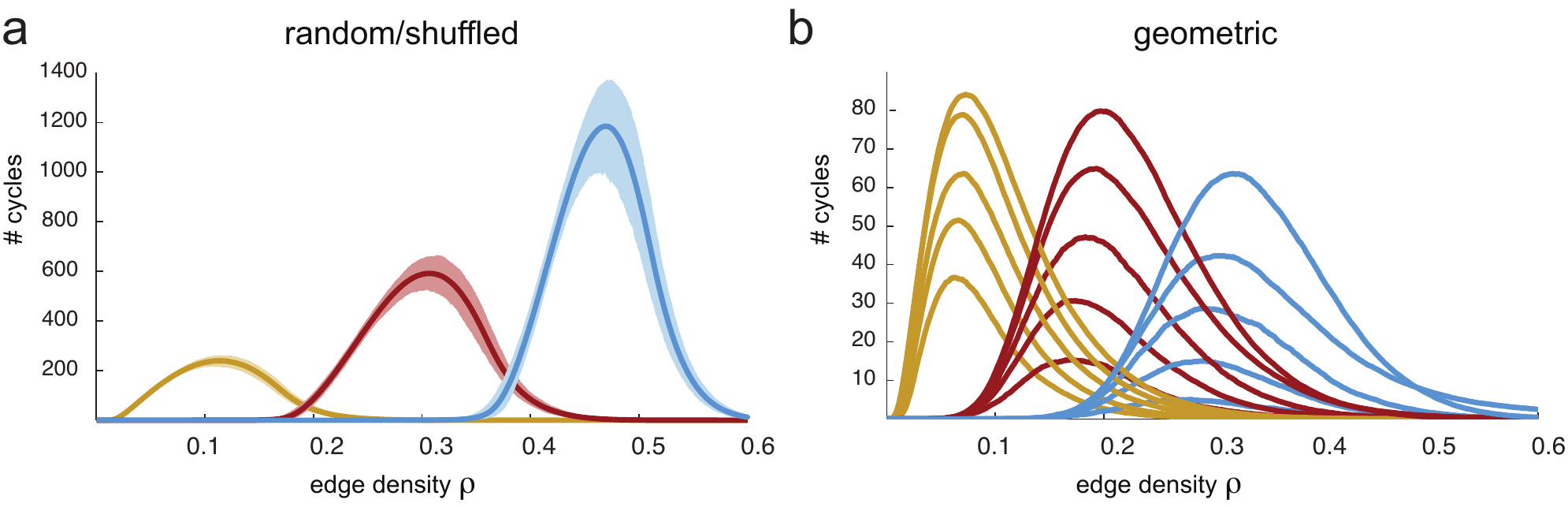}
\caption{Betti curves for random and geometric matrices.  (a) N = 100, and means for $\beta_1(\rho)$ (yellow), $\beta_2(\rho)$ (red), and $\beta_3(\rho)$ (blue) are displayed with bold lines, while shading indicates 99.5\% confidence intervals. (b) N=100, and average Betti curves are displayed for dimensions $d=10, 50, 100, 1000, 10000,$ in increasing order (i.e., higher curves correspond to larger dimensions).}
  \end{center}
  \end{figure}
 \vspace{-.2in}

A \emph{geometric order complex} is one arising from the negative distance matrix of a collection of points embedded in some Euclidean space.  We choose {\it negative} distance matrices so that the highest matrix values correspond to the nearest distances; this is consistent with the intuition that correlations should decrease with distance, as described in the main text.  Sampling such a complex consists of sampling $N$ i.i.d. points, $\{p_i\}$,  from some distribution on $\mathbb{R}^d$.  The associated sequence of clique complexes, $X(G_0) \subset X(G_1) \subset \ldots \subset X(G_{p+1}),$ corresponding to a geometric order complex 
is also referred to as the Vietoris-Rips complex of the underlying points.  

These complexes have been heavily studied in cases where the points are presumably sampled from an underlying manifold \cite{CollinsZCG04}. In our setting, however, we sample points from the uniform distribution on the unit cube in $\mathbb{R}^d$ for $d \leq N.$  To our knowledge, the Betti curves of geometric order complexes are largely unstudied. Our numerical experiments show that they are highly stereotyped (Figure 6b), irrespective of $d$ for a large range of dimensions.\footnote{We observed similar Betti curves to those in Figure 6b for values of $d$ that were orders of magnitude larger than $N$.  Nevertheless, there is some evidence to indicate that Betti curves will approach those of random order complexes as $d \rightarrow \infty$.}
Moreover, they are roughly an order of magnitude smaller at the peak than the Betti curves of random order complexes with matching $N$, and the peak values {\it decrease} rather than increase as we move between $\beta_1(\rho)$ to $\beta_2(\rho)$ and $\beta_3(\rho)$.
The differences between the Betti curves of random and geometric matrices can also be understood through the lens of \emph{persistence lifetimes}, which we will describe in section \ref{S:persistence}.

\subsection{Dimension of geometric order complexes}
 
Any matrix ordering $\widehat{M}$ appears with equal probability in the distribution of random symmetric matrices with i.i.d. entries.   The consistency of the Betti curves in Figure 6a indicates that ``most'' of these matrix orderings have a similar organization of cliques.  For geometric matrices, the possible matrix orderings are sampled in a highly non-uniform manner, leading to dramatically different Betti curves.
Despite this, it is worth noting that any matrix ordering can in fact arise from a distance matrix.

\begin{defn} A set of points $p_1,\ldots,p_N\in \mathbb{R}^d$ is called a {\it geometric realization} of the matrix ordering $\widehat{M}$ if the distance matrix $D_{ij}=||p_i-p_j||$ has $\widehat{D} = \widehat{M}$.
\end{defn}

Note that for each collection of three or more points, the (higher) triangle inequalities implied by the metric impose strong constraints on $\widehat{M}$. 
This means that for most matrix orderings, the probability of sampling a point configuration in the unit cube that yields a geometric realization of $\widehat{M}$ is vanishingly small.  This is why geometric Betti curves are, on average, so different from those of random matrices.  Nevertheless, geometric realizations do always exist, provided $d \geq N-1$. 

\begin{lem} Every $N\times N$ matrix ordering $\widehat{M}$ that has ${N \choose 2}$ distinct off-diagonal entries possesses  a geometric realization in $(N-1)$-dimensional Euclidean space. Moreover, this realization can be chosen as 
\begin{equation*} p_i=\frac1{\sqrt{2}}\left(\vec  e_i  - \dfrac{\varepsilon}{2} \sum_{j=1}^N M_{ij}\vec  e_j\right),
\end{equation*}
for small enough $\varepsilon>0,$ where $M$ is any symmetric matrix with ordering $\widehat{M}$ and zeroes on the diagonal, and
 $\{\vec  e_i\}_{i=1}^N$ is the standard orthonormal basis in $\mathbb R^N$.
\end{lem} 

\begin{proof} 
With the choice above, $||p_i-p_j||^2 = ||p_i||^2 + ||p_j||^2 - 2 p_i \cdot p_j = 1+\varepsilon M_{ij} + \mathcal{O}(\varepsilon^2).$ 
\end{proof}

Despite this fact, when we constrain the dimension $d$ of the Euclidean space we do find matrix orderings that cannot be geometrically realized at all.  This was the basis for our examples in Figure 2a of the main text.

\subsection{Figure 2a examples from the main text}   Here we prove that the $d \geq 2$ and $d \geq 3$ matrices (reproduced in Figure 1) cannot be geometrically realized in lower dimensions.

To see why the $d \geq 2$ matrix cannot arise from an arrangement of points on a line, observe that the three smallest matrix entries are $M_{12}, M_{13},$ and $M_{14}$.  This implies the three shortest distances in a corresponding point arrangement must all involve the point $p_1$, which is not possible for points on a line.  

To see why the $d\geq 3$ matrix cannot arise from an arrangement of points on a plane, notice that the six smallest matrix entries are $M_{i\alpha}$, for $i = 1,2,3$ and $\alpha=4,5$.  This means the six smallest distances are those of the form $\Vert p_i - p_\alpha \Vert$, for $i = 1,2,3$ and $\alpha=4,5$.  Without loss of generality we can assume $\Vert p_i - p_\alpha \Vert < 1$, and all other distances are greater than one.  Now suppose the points $p_1,\ldots,p_5$ all lie in a plane.  Then $p_4, p_5 \in D(p_1) \cap D(p_2) \cap D(p_3)$, where $D(p_i)$ is a disk of radius $1$ centered at $p_i$.   Since none of the disk centers in contained in any of the other two disks, the largest distance between two points in the intersection $D(p_1) \cap D(p_2) \cap D(p_3)$ is less than one, and thus $\Vert p_4-p_5\Vert < 1$, which is a contradiction.  We conclude that the matrix cannot arise from points in the plane.  We thank Anton Petrunin for this example.

\section{Computational aspects and persistence}\label{S:computation}
  
Each graph in an order complex, $G_0 \subset G_1 \subset \cdots \subset G_{p+1},$ is a subgraph of its successor.  Intuitively, this means that the clique topology of any $G_{r}$ is closely related to the clique topology of the previous graph, $G_{r-1}$.  Exploiting this structure dramatically reduces the computational complexity of finding Betti curves (defined in section~\ref{sec:clique-top-across}), and also provides us with finer matrix invariants in the form of \emph{persistence lifetimes} of cycles.  This is achieved via {\it persistent homology}, an approach that enables homology cycles to be tracked as we move from one graph in the order complex to the next.  

\subsection{A brief history of persistent homology}

The mathematics underlying persistent homology has existed since the middle of the twentieth century, in the guise of Morse theory and spectral sequences for the homology of filtered spaces.  Its interpretation as a tool for data analysis, however, is a much more recent development. One can trace the origins of these applications to work on size theory in computer vision \cite{FrosiniLandi, CagliariFerriPozzi, FrosiniMulazzani} and alpha shapes in computational geometry \cite{Robins, Edelsbrunner, EdelsbrunnerMucke, EdelsbrunnerLetscherZomorodian}. The use of persistent homology as a tool for the study of data sets relies on two fundamental and recent developments: computabilty and robustness.

Computability arose from the \emph{persistence algorithm}, developed first for subsets of three-dimensional complexes in \cite{EdelsbrunnerLetscherZomorodian} and then extended to work with general simplicial complexes in \cite{ZomorodianCarlsson2005}. In addition to the algorithm, these papers introduced the notions of persistence diagrams and modules. Several software packages \cite{javaPlex, Dionysus, phom} have been developed based on the persistence algorithm, and recent work using discrete Morse theory has led to further improvements in speed and memory efficiency \cite{Perseus}. 

Robustness to perturbations of the underlying simplicial complexes, on the other hand, was first explicitly shown through the bottleneck stability theorem of \cite{Cohen-SteinerEdelsbrunnerHarer}. Further work has broadened this result by developing more complete theoretical tools for the comparison of persistence structures, divorcing their stability from any underlying geometry \cite{ChazalCohenSteinerGlisseGuibasOudot, ChazaldeSilvaGlisseOudot, ChazaldeSilvaOudot}. It is this interpretation of stability that most clearly applies to our study of order complexes.

Although persistent homology has only recently emerged a tool for studying features of data, it has already found a broad range of applications \cite{ChanCarlssonRabadan, SinghMemoliIshkhanovSapiroCarlssonRingach, GambleHeo, CarlssonIshkhanovdeSilvaZomorodian}.

\subsection{Persistent homology of order complexes} \label{S:persistence}
Here we present the basic ideas in persistent homology, restricted to the special case of computing clique topology for order complexes.  This means we need to apply the persistence algorithm to filtered families of clique complexes,
$$X(G_0) \subset X(G_1) \subset \cdots \subset X(G_p) \subset X(G_{p+1}),$$
where the graphs $\{G_r\}$ comprise the order complex of a symmetric matrix.
In order to track homology cycles from one clique complex to the next, we need to understand how the natural inclusion maps on the graphs, $\iota_r:G_r \into G_{r+1},$ translate to maps on the corresponding cliques, chains, and homology groups, $H_m(X(G_r))$.  This turns out to be straightforward, as there is an obvious extension to maps on clique complexes, $\iota_r: X(G_r) \into X(G_{r+1}),$ and these in turn can be extended linearly to maps between chain groups.
   
  \begin{lem} Consider the order complex  $G_0 \subset G_1 \subset \cdots \subset G_{p+1}.$
  The standard inclusion maps, $\iota_{r} : G_r \into G_{r+1}$, induce maps on homology $(\iota_{r})_m : H_m(X(G_r)) \to H_m(X(G_{r+1}))$.
  \end{lem}
  
Using these maps, we can follow individual cycles and understand their evolution as we move from one graph to the next in the order complex.  Of particular interest are the edge densities at which cycles appear and disappear (Figure 7). 

  \begin{defn}
  Let $\omega \in H_m(X(G_r))$ be a non-zero cycle which is not in the image of $\iota_{r-1}$, and let $s>r$ be the smallest integer such that $\iota_{s-1} \circ \iota_{s-2} \circ \ldots \circ \iota_r(\omega) = 0.$ 
We say that $\omega$ is \emph{born} at $r$ and \emph{dies} at $s$, and has \emph{persistence lifetime} $\ell(\omega) = s-r$.
  \end{defn}

  \begin{figure}[!h]
\begin{center}
\includegraphics[width=3.5in]{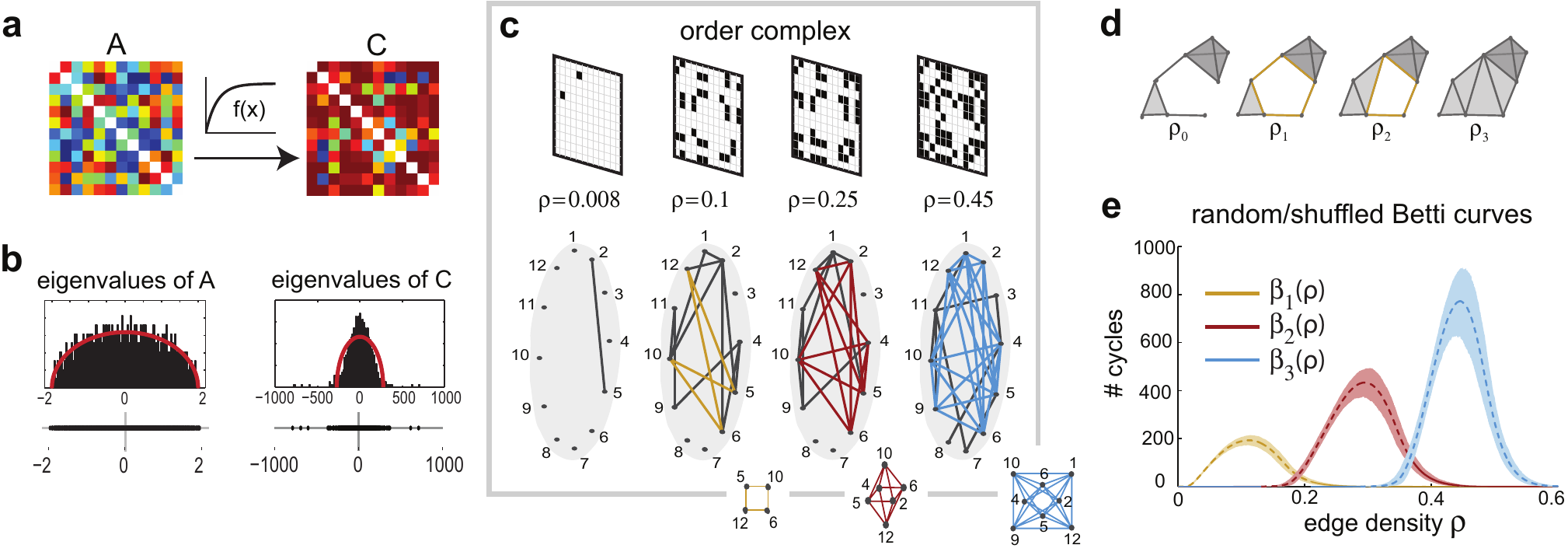}
\caption{Illustration of persistence lifetime. The $1$-cycle (yellow) appears at edge density $\rho_1$ and disappears at $\rho_3$, so it's lifetime is $\ell = \rho_3-\rho_1$.}
  \end{center}
  \end{figure}

For a given order complex, the distribution of persistence lifetimes provides a measure of matrix structure that is complementary to the Betti curves defined in section~\ref{sec:clique-top-across}.

\subsection{Persistence lifetimes of random and geometric order complexes}

Recall that there is a sharp qualitative difference in the Betti curves of random order complexes and those of geometric order complexes (Figure 6). These differences are also reflected in the distributions of their persistence lifetimes.  While random complexes have relatively broad distributions (Figure 8a), the geometric complexes are heavily weighted toward shorter lifetimes (Figure 8b). The shapes of these distributions are a direct consequence of the order in which edges are added in the order complex. 

  \begin{figure}[!h]
\begin{center}
\includegraphics[width=6in]{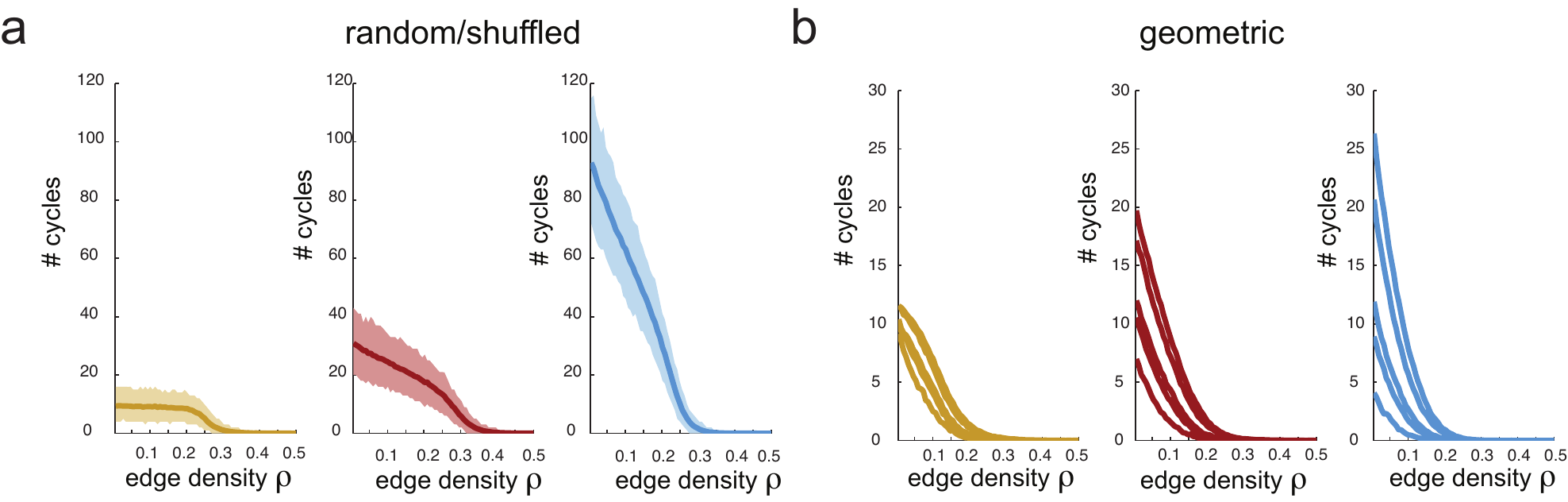}
\caption{Persistence lifetimes for random and geometric matrices.  (a) N = 100, and mean lifetime distributions for $1$-cycles (yellow), $2$-cycles (red), and $3$-cycles (blue) are displayed with bold lines, while shading indicates 99.5\% confidence intervals. (b) N=100, and average lifetime distributions are displayed for dimensions $d=10, 50, 100, 1000, 10000,$ in increasing order (i.e., higher curves correspond to larger dimensions).}
  \end{center}
  \end{figure}

The qualitative differences in these distributions can be understood by 
thinking about dependencies in edge orderings in the order complex. 
Minimal cycles, represented by cross-polytopes (Figure 5), are known to constitute 
the large majority of cycles in random order complexes \cite{KahleMeckes2013}, and 
can thus be used to understand the shape of the distribution. Such a cycle's 
lifetime is governed by the density at which the first additional edge appears, since the extra edge
destroys the cycle by creating new cliques.
Since the ordering of the edges is completely random, the lifetimes will be broadly 
distributed. In contrast, geometric order complexes are constrained by
triangle inequalities (and higher-dimensional analogues); these produce dependencies in the edge ordering 
which imposes an upper limit on the lifetime of small cycles, like the cross-polytopes.\footnote{This is true statistically.  It is, of course, possible for two distances to be close in absolute terms and still be separated by many edges in the order complex, but this is rare enough that the intuition about Betti curves still holds.}  Persistence lifetimes in geometric complexes are thus concentrated at short lifetimes.

\subsection{CliqueTop software}\label{S:software}

To compute clique topology for symmetric matrices, we developed the CliqueTop Matlab package. This software is maintained by Chad Giusti, and is available on GitHub at \url{https://github.com/nebneuron/clique-top}.  At the time of this writing, CliqueTop makes use of one other package: Perseus \cite{Perseus}, by Vidit Nanda.  Perseus provides an implementation of the 
persistence algorithm, and is available at \url{http://www.sas.upenn.edu/~vnanda/perseus/index.html}.  
Previous versions of CliqueTop also used the Cliquer software package \cite{Cliquer}.

\bibliography{supp_text_references}

\bibliographystyle{alpha}

\end{document}